\documentclass[11pt]{article}
\usepackage{amsxtra}
\usepackage{amssymb, amsmath}
\usepackage{amsthm}
\usepackage {hyperref}

\newtheorem{theorem}{Theorem}[section]

\newtheorem{corollary}{Corollary}[section]
\newtheorem{result}{Result}[section]

\newtheorem{example}{Example}[section]

\usepackage[dvips]{graphicx}

\begin{document}
\begin{center}
\textbf{\LARGE{Generalized Symmetric Divergence Measures and the Probability of Error}}
\end{center}

\smallskip
\begin{center}
\textbf{\large{Inder Jeet Taneja}}\\
Departamento de Matem\'{a}tica\\
Universidade Federal de Santa Catarina\\
88.040-900 Florian\'{o}polis, SC, Brazil.\\
\textit{e-mail: taneja@mtm.ufsc.br\\
http://www.mtm.ufsc.br/$\sim $taneja}
\end{center}

\begin{abstract}
\textit{There are three classical divergence measures exist in the literature on
information theory and statistics. These are namely, Jeffryes-Kullback-Leiber  \cite{jef}, \cite{kul} J-divergence. Sibson-Burbea-Rao  \cite{sib}, \cite{bur} Jensen-Shannon divegernce and Taneja \cite{tan2} Arithmetic-Geometric divergence. These three measures bear an interesting relationship among each other. The divergence measures like Hellinger \cite{hel} discrimination,  symmetric $\chi ^2 - $  divergence, and triangular discrimination are also known in the literature. In this paper, we have considered generalized symmetric divergence measures having the measures given above as particular cases. Bounds on the probability of error are obtained in terms of generalized symmetric divergence measures. Study of bounds on probability of error is extended for the difference of divergence measures.}
\end{abstract}



\section{Introduction}

Let
\[
\Gamma _n = \left\{ {P = (p_1 ,p_2 ,...,p_n )\left| {p_i > 0,\sum\limits_{i
= 1}^n {p_i = 1} } \right.} \right\},
\,
n \ge 2,
\]

\noindent
be the set of all complete finite discrete probability distributions. For all $P,Q \in \Gamma _n $, the following measures are well known in the literature on information theory and statistics:

\smallskip
Let us consider the measure
\begin{equation}
\label{eq1}
\zeta _s (P\vert \vert Q) = \begin{cases}
 {J_s (P\vert \vert Q) = \left[ {s(s - 1)} \right]^{ - 1}\left[
{\sum\limits_{i = 1}^n {\left( {p_i^s q_i^{1 - s} + p_i^{1 - s} q_i^s }
\right) - 2} } \right],} & {s \ne 0,1} \\
 {J(P\vert \vert Q) = \sum\limits_{i = 1}^n {\left( {p_i - q_i } \right)\ln
\left( {\frac{p_i }{q_i }} \right),} } & {s = 0,1} \\
\end{cases}
\end{equation}

\noindent
for all $P,Q \in \Gamma _n $

\smallskip
The measure (\ref{eq1}) can be seen in Burbea and Rao \cite{bur} and Taneja \cite{tan2, tan3}. These authors studied the above measure using the multiplicative constant as $(s - 1)^{ - 1}$, while Taneja \cite{tan4} studied considering it as $\left[ {s(s -1)} \right]^{ - 1},\;s \ne 0,1$.

\smallskip
The expression (\ref{eq1}) admits the following particular cases:

(i) $\zeta _{ - 1} (P\vert \vert Q) = \zeta _2 (P\vert \vert Q) = \frac{1}{2}\Psi (P\vert \vert Q)$,

(ii) $\zeta _0 (P\vert \vert Q) = \zeta _1 (P\vert \vert Q) = J(P\vert \vert Q)$,

(iii) $\zeta _{1 / 2} (P\vert \vert Q) = 8\,h(P\vert \vert Q)$,

\smallskip
\noindent where
\begin{equation}
\label{eq2}
\Psi (P\vert \vert Q) = \chi ^2(P\vert \vert Q) + \chi ^2(Q\vert \vert P) =
\sum\limits_{i = 1}^n {\frac{(p_i - q_i )^2(p_i + q_i )}{p_i q_i }} ,
\end{equation}

\noindent and
\begin{equation}
\label{eq3}
h(P\vert \vert Q) = \frac{1}{2}\sum\limits_{i = 1}^n
{(\sqrt {p_i } - \sqrt {q_i } )^2} \, ,
\end{equation}

\noindent
are the \textit{symmetric }$\chi ^2$\textit{-- divergence} \cite{tan4, tan5} and \textit{Hellingar's discrimination} \cite{hel} respectively. More details on these divergence measures can be seen in  Taneja \cite{tan4}.

\smallskip
Let us consider now the another generalized measure
\begin{equation}
\label{eq4}
\xi _s (P\vert \vert Q) =
\begin{cases}
 {IT_s (P\vert \vert Q) = \left[ {s(s - 1)} \right]^{ - 1}\left[
{\sum\limits_{i = 1}^n {\left( {\frac{p_i^s + q_i^s }{2}} \right)\left(
{\frac{p_i + q_i }{2}} \right)} ^{1 - s} - 1} \right],} & {s \ne 0,1} \\
 {I(P\vert \vert Q) = \frac{1}{2}\left[ {\sum\limits_{i = 1}^n {p_i \ln
\left( {\frac{2p_i }{p_i + q_i }} \right) + \sum\limits_{i = 1}^n {q_i \ln
\left( {\frac{2q_i }{p_i + q_i }} \right)} } } \right],} & {s = 1} \\
 {T(P\vert \vert Q) = \sum\limits_{i = 1}^n {\left( {\frac{p_i + q_i }{2}}
\right)\ln \left( {\frac{p_i + q_i }{2\sqrt {p_i q_i } }} \right)} ,} & {s =
0} \\
\end{cases}
\end{equation}

\noindent
for all $P,Q \in \Gamma _n $

\smallskip
The measure (\ref{eq4}) is new in the literature and is studied for the first time by Taneja \cite{tan2}  and is called \textit{arithmetic and geometric mean divergence measure.} More details on these divergence measures can be seen in Taneja \cite{tan2, tan3}.

\smallskip
The measure (\ref{eq4}) admits the following particular cases:

(i) $\xi _{ - 1} (P\vert \vert Q) = \frac{1}{4}\Delta (P\vert \vert Q)$.

(ii) $\xi _1 (P\vert \vert Q) = I(P\vert \vert Q)$.

(iii) $\xi _{1 / 2} (P\vert \vert Q) = 4\;d(P\vert \vert Q)$.

(iii) $\xi _0 (P\vert \vert Q) = T(P\vert \vert Q)$.

(iv) $\xi _2 (P\vert \vert Q) = \frac{1}{16}\Psi (P\vert \vert Q)$.

\smallskip
\noindent where
\begin{equation}
\label{eq5}
\Delta (P\vert \vert Q) = \sum\limits_{i = 1}^n {\frac{(p_i - q_i )^2}{p_i + q_i }} ,
\end{equation}

\noindent and
\begin{equation}
\label{eq6}
d(P\vert \vert Q) = 1 - \sum\limits_{i = 1}^n {\left( {\frac{\sqrt {p_i } +
\sqrt {q_i } }{2}} \right)} \left( {\sqrt {\frac{p_i + q_i }{2}} } \right).
\end{equation}

\noindent
are the \textit{triangular discrimination} and \textit{d-divergence} respectively.

\smallskip
The symmetric divergence measures (\ref{eq1}) and (\ref{eq4}) admit several particular
cases. An inequality among these measures \cite{tan4} is given by
\begin{equation}
\label{eq7}
\frac{1}{4}\Delta (P\vert \vert Q) \le I(P\vert \vert Q) \le h(P\vert \vert
Q) \le 4\,d(P\vert \vert Q)
 \le \frac{1}{8}J(P\vert \vert Q) \le T(P\vert \vert Q) \le \frac{1}{16}\Psi
(P\vert \vert Q).
\end{equation}

In this paper our aim is to obtain the bounds on the probability of error in terms of the measure (\ref{eq1}) and (\ref{eq4}) and then establish examples of the seven measures given in (\ref{eq7}). Some studies on divergence measures and probability of error can be seen in Taneja \cite{tan1, tan2}.

\section{Error Probability for Two Class Case}

Suppose we have two pattern classes $C = \{C_1 ,C_2 \}$ with \textit{``a priori''} probability
$p_i = \Pr \{C = C_i \}$, $i = 1,2$. Let the feature $x$ on $X$ have a
class-conditional probability density function $P(x\vert C_i )$ are known.
Given a feature $x$ on $X$, we can calculate the conditional (\textit{a posteriori}) probability
$P(C_i \vert x)$ for each $i$, by the Bayes rule:
\[
P(C_i \vert x) = \Pr \{C = C_i \vert X = x\} = \frac{p_i P(x\vert C_i
)}{\sum\limits_{k = 1}^2 {p_k P(x\vert C_k )} },
\,
i = 1,2.
\]

It is well known that the decision rule which minimizes the probability of
error is the Bayes decision rule which chooses the hypotheses (pattern
classes) with the largest posterior probability. Using this rule, the
partial probability of error for given $X = x$ is expressed by
\[
P(e\vert x) = 1 - \max \left\{ {P(C_1 \vert x),P(C_2 \vert x)} \right\} =
\min \left\{ {P(C_1 \vert x),P(C_2 \vert x)} \right\}.
\]

Prior to observing $X$, the probability of error $P_e $, associated with $X$
is defined as the expected probability of error, i.e.,
\[
P_e = E_X \left\{ {P(e\vert x)} \right\} = \int\limits_X {P(e\vert x)p(x)dx},
\]
\noindent
where $p(x) = \sum\limits_{i = 1}^2 {p_i P(x\vert C_i )} $ is the
unconditional density of $X$evaluated as $x$.

\smallskip
In terms of the prior probabilities for the two class case, let us consider
generalized divergence measures:
\begin{equation}
\label{eq8}
\zeta _s = \begin{cases}
 {J_s ,} & {s \ne 0,1} \\
 {J,} & {s = 0,1} \\
\end{cases},
\end{equation}

\noindent where
\[
J_s = \left[ {s(s - 1)} \right]^{ - 1}\left[ {\int {\left( {p(x\vert C_1
)^sp(x\vert C_2 )^{1 - s} + p(x\vert C_2 )^sp(x\vert C_1 )^{1 - s}}
\right)dx - 2} } \right],
\,
s \ne 0,1
\]

\noindent and
\[
J = \int {\left( {P(x\vert C_1 ) - P(x\vert C_2 )} \right)\log \left(
{\frac{P(x\vert C_1 )}{P(x\vert C_2 )}} \right)} dx.
\]

Let us consider now
\begin{equation}
\label{eq9}
\xi _s = \begin{cases}
 {IT_s ,} & {s \ne 0,1} \\
 {I,} & {s = 1} \\
 {T,} & {s = 0} \\
\end{cases},
\end{equation}

\noindent where
\[
IT_s = \left[ {s(s - 1)} \right]^{ - 1}\left[ {\int {\left( {\frac{p(x\vert
C_1 )^s + p(x\vert C_2 )^s}{2}} \right)} \left( {\frac{p(x\vert C_1 ) +
p(x\vert C_1 )}{2}} \right)^{1 - s}dx - 1} \right],
\, s \ne 0,1
\]
\[
I = \int {\left[ {\frac{P(x\vert C_1 )}{2}\ln \left( {\frac{2P(x\vert C_1
)}{P(x\vert C_1 ) + P(x\vert C_2 )}} \right) + \frac{P(x\vert C_2 )}{2}\ln
\left( {\frac{2P(x\vert C_2 )}{P(x\vert C_1 ) + P(x\vert C_2 )}} \right)}
\right]} dx,
\]

\noindent and
\[
T = \int {\left( {\frac{P(x\vert C_1 ) + P(x\vert C_2 )}{2}} \right)} \ln
\left( {\frac{P(x\vert C_1 ) + P(x\vert C_2 )}{2\sqrt {P(x\vert C_1
)P(x\vert C_2 )} }} \right)dx.
\]

Let us define the above measures in more general way
\[
\zeta _s (p_1 ,p_2 ) = \begin{cases}
 {J_s (p_1 ,p_2 ),} & {s \ne 0,1} \\
 {J,(p_1 ,p_2 )} & {s = 0,1} \\
\end{cases},
\]

\noindent where
\begin{align}
J_s (p_1 ,p_2 ) = \left[ {s(s - 1)} \right]^{ - 1} & \left\{ {\int {\left[
{\left( {p_1 p(x\vert C_1 )} \right)^s\left( {p_2 p(x\vert C_2 )} \right)^{1
- s}} \right.} } \right.\notag\\
& \left. {\left. { + \left( {p_2 p(x\vert C_2 )} \right)^s\left( {p_1 p(x\vert
C_1 )} \right)^{1 - s}} \right]dx - 1} \right\} ,
\, s \ne 0,1\notag
\end{align}

\noindent and
\[
J = \int {\left( {p_1 P(x\vert C_1 ) - p_2 P(x\vert C_2 )} \right)\ln \left(
{\frac{p_1 P(x\vert C_1 )}{p_2 P(x\vert C_2 )}} \right)} dx.
\]

Similarly let us define
\[
\xi _s (p_1 ,p_2 ) = \begin{cases}
 {IT_s (p_1 ,p_2 ),} & {s \ne 0,1} \\
 {I(p_1 ,p_2 ),} & {s = 1} \\
 {T(p_1 ,p_2 ),} & {s = 0} \\
\end{cases},
\]

\noindent where
\begin{align}
IT_s (p_1 ,p_2 ) & = \left[ {s(s - 1)} \right]^{ - 1}  \left\{ {\int {\left[
{\left( {\frac{\left( {p_1 p(x\vert C_1 )} \right)^s + \left( {p_2 p(x\vert
C_2 )} \right)^s}{2}} \right)} \right.} } \right. \cdot \notag\\
& \hspace{120pt}  \cdot \left. {\left. {\left( {\frac{p_1 p(x\vert C_1 ) + p_2 p(x\vert C_1
)}{2}} \right)^{1 - s}dx} \right] - \frac{1}{2}} \right\} ,
\notag\\
I (p_1 ,p_2 )& =  \frac{1}{2} \left\{ {\int {\left[ {p_1 P(x\vert C_1 ) \ln \left(
{\frac{2p_1 P(x\vert C_1 )}{p_1 P(x\vert C_1 ) + p_1 P(x\vert C_2 )}}
\right)} \right.} } \right.\notag\\
& \hspace{100pt} \left. {\left. { + \,p_2 P(x\vert C_2 )\ln \left( {\frac{2p_2 P(x\vert C_2
)}{p_1 P(x\vert C_1 ) + p_1 P(x\vert C_2 )}} \right)} \right]dx} \right\} \notag\\
\intertext{and}
T(p_1 ,p_2 ) &  = \int {\left( {\frac{p_1 P(x\vert C_1 ) + p_2 P(x\vert C_2 )}{2}}
\right)} \ln \left( {\frac{p_1 P(x\vert C_1 ) + p_2 P(x\vert C_2 )}{2\sqrt
{p_1 p_2 P(x\vert C_1 )P(x\vert C_2 )} }} \right)dx.\notag
\end{align}

It is easy to verify
\[
\zeta _s \left( {\frac{1}{2},\frac{1}{2}} \right) = \left( {\frac{1}{2}}
\right)\zeta _s ,
\]

\noindent and
\[
\xi _s \left( {\frac{1}{2},\frac{1}{2}} \right) = \left( {\frac{1}{2}}
\right)\xi _s ,
\]

\noindent where $\zeta _s $ and $\xi _s $ are as given by (\ref{eq8}) and (\ref{eq9})
respectively.

\smallskip
Using Bayes rules, $p_1 P(x\vert C_1 ) = p(x)P(C_1 \vert x)$ and $p_2 P(x\vert C_2 ) = p(x)P(C_2 \vert x)$, we can write
\[
\zeta _s (p_1 ,p_2 ) = \int {p(x)} \zeta _s (x)dx
\]

\noindent and
\[
\xi _s (p_1 ,p_2 ) = \int {p(x)} \xi _s (x)dx,
\]

\noindent where
\begin{equation}
\label{eq10}
\zeta _s (x) = \begin{cases}
 {\begin{array}{l}
 J_s (x) = [s(s - 1)]^{ - 1}\left[ {P(C_1 \vert x)^sP(C_2 \vert x)^{1 - s}}
\right. \\
 \left. {\, \, \, \, \;\, \, + \,P(C_2 \vert x)^sP(C_2
\vert x)^{1 - s} - 1} \right], \\
 \end{array}} & {\begin{array}{l}
 \\
 s \ne 0,1 \\
 \end{array}} \\
 {J(x) = \left( {P(C_1 \vert x) - P(C_2 \vert x)} \right)\ln \left(
{\frac{P(C_1 \vert x)}{P(C_2 \vert x)}} \right),} & {s = 0,1} \\
\end{cases}
\end{equation}

\noindent and
\[
\xi _s (x) = \begin{cases}
 {IT_s (x) = \left[ {s(s - 1)} \right]^{ - 1}\left[ {\left( {\frac{P(C_1
\vert x)^s + P(C_2 \vert x)^s}{2}} \right)\left( {\frac{1}{2}} \right)^{1 -
s} - \frac{1}{2}} \right],} & {s \ne 0,1} \\
 {\begin{array}{l}
 I(x) = \frac{1}{2}\left[ {P(C_1 \vert x)\ln \left( {\frac{2P(C_1 \vert
x)}{P(C_1 \vert x) + P(C_2 \vert x)}} \right)} \right. \\
 \left. {\, \, \, \, + \,\,P(C_2 \vert x)\ln \left(
{\frac{2P(C_2 \vert x)}{P(C_1 \vert x) + P(C_2 \vert x)}} \right)} \right],
\\
 \end{array}} & {\begin{array}{l}
 \\
 \\
 s = 1 \\
 \end{array}} \\
 {T(x) = \left( {\frac{P(C_1 \vert x) + P(C_2 \vert x)}{2}} \right)\ln
\left( {\frac{P(C_1 \vert x) + P(C_2 \vert x)}{2\sqrt {P(C_1 \vert x)P(C_2
\vert x)} }} \right),} & {s = 0} \\
\end{cases}.
\]

\subsection{ Bounds on Generalized Divergence Measures}

In this paper we shall obtain error bounds in terms of these generalized
divergence measures. Before that we review some known results.

\smallskip
There exist in the literature lower bounds on $P_e $ in terms of different
divergence measures. For review refer to Taneja \cite{tan1, tan2}.
Kailath \cite{kai} gave the following bound:
\[
P_e \ge \frac{1}{4}\exp \left( { - \frac{J}{2}} \right),
\]

\noindent
where the equality holds for $J = \infty $. Later, Toussaint \cite{tou1, tou2}  gave a
tighter and general bound as
\[
P_e \ge \frac{1}{2} - \frac{1}{2}\sqrt {1 - 4\exp \left[ { - 2\,H(p_1 ,p_2 )
- J(p_1 ,p_2 )} \right]} ,
\]

\noindent where
\[
H(p_1 ,p_2 ) = - p_1 \log p_1 - p_2 \log p_2 .
\]

For $p_1 = p_2 = \frac{1}{2}$, we have
\[
P_e \ge \frac{1}{2} - \frac{1}{2}\sqrt {1 - 4\exp \left( { - \frac{J}{2}}
\right)} ,
\]

\noindent
where the equality hold for both $J = 0$ and $J = \infty $.

\smallskip
Toussaint \cite{tou1, tou2}  still obtain a sharper bound on $J$ in terms of $P_e $ as
\[
J(p_1 ,p_2 ) \ge (2P_e - 1)\log \left( {\frac{P_e }{1 - P_e }} \right).
\]

When$p_1 = p_2 = \frac{1}{2}$, one gets
\[
J \ge 2(2P_e - 1)\log \left( {\frac{P_e }{1 - P_e }} \right).
\]

\begin{theorem} We have
\[
\zeta _s (p_1 ,p_2 ) \ge \zeta _s (P_e ),
\]

\noindent where
\[
\zeta _s (P_e ) = \begin{cases}
 {J_s (P_e ) = [s(s - 1)]^{ - 1}\left[ {P_e ^s(1 - P_e )^{1 - s} + (1 - P_e
)^sP_e ^{1 - s} - 1} \right]} & {s \ne 0,1} \\
 {J(P_e ) = \left( {2P_e - 1} \right)\ln \left( {\frac{P_e }{1 - P_e }}
\right),} & {s = 0,1} \\
\end{cases}.
\]
\end{theorem}

\begin{proof} It is sufficient to show the result only for $s \ne 0,1$, since for $s = 0,1$ the result is already known in the literature or follows in view of continuity with respect parameter $s$.

\smallskip
We have
\[
P_e (x) = \min \left\{ {P(C_1 \vert x),\,P(C_2 \vert x)} \right\}.
\]

As $\zeta _s (x)$ given by (\ref{eq10}) is symmetric with respect to $P(C_1 \vert
x)$ and $P(C_2 \vert x)$, let us consider $P(C_1 \vert x) = P_e (x)$ and
$P(C_2 \vert x) = 1 - P_e (x)$. Making these substitutions in (\ref{eq10}), we get
\begin{equation}
\label{eq11}
\zeta _s (x) = \begin{cases}
 {J_s (x) = [s(s - 1)]^{ - 1}\left[ {P_e (x)^s\left( {1 - P_e (x)}
\right)^{1 - s}\left( {1 - P_e (x)} \right)^sP_e (x)^{1 - s} - 1} \right],}
& {s \ne 0,1} \\
 {J(x) = \left( {2P_e (x) - 1} \right)\ln \left( {\frac{P_e (x)}{1 - P_e
(x)}} \right),} & {s = 0,1} \\
\end{cases}.
\end{equation}

Now consider a function
\[
f_s (a) = \begin{cases}
 {[s(s - 1)]^{ - 1}\left[ {a^s\left( {1 - a} \right)^{1 - s} + \,\left( {1 -
a)} \right)^sa^{1 - s} - 1} \right],} & {s \ne 0,1} \\
 {\left( {2a - 1} \right)\ln \left( {\frac{a}{1 - a}} \right),} & {s = 0,1}
\\
\end{cases}
\]

\noindent or, simply,
\[
f_s (a) = [s(s - 1)]^{ - 1}\left[ {a^s\left( {1 - a} \right)^{1 - s} +
\,\left( {1 - a)} \right)^sa^{1 - s} - 1} \right],\;s \ne 0,1.
\]

We have
\begin{align}
{f}'_s (a) = \left[ {s(s - 1)a(1 - a)} \right]^{ - 1} & \left[ {(s - 1)a^s}
\right.(1 - a)^{1 - s} \notag\\
& \left. { + \,(1 - s)a^{1 - s}(1 - a) - (1 - a)^sa^{2 - s}} \right]\notag
\end{align}

\noindent and
\[
{f}''_s (a) = a^{s - 2}(1 - a)^{ - s - 1} + (1 - a)^{s - 2}a^{ - s - 1}
\]

\noindent
for all $0 < a < \textstyle{1 \over 2}$.

\smallskip
This gives${f}''_s (a) \ge 0$, for any$s \ne 0,1$. For $s = 0,1$ the result
follows by continuity, i.e, $f_s (a)$ is convex function of $a$ ($0 < a <
\textstyle{1 \over 2})$ for any $s \in ( - \infty ,\infty )$.

\smallskip
Taking expectation on both sides of (\ref{eq11}) and using Jensen's inequality we
get the required result.
\end{proof}

\begin{corollary} When $p_1 = p_2 = \frac{1}{2}$, we have
\[
\zeta _s \ge \begin{cases}
 {2[s(s - 1)]^{ - 1}\left[ {P_e ^s(1 - P_e )^{1 - s} + (1 - P_e )^sP_e ^{1 -
s} - 1} \right]} & {s \ne 0,1} \\
 {2\left( {2P_e - 1} \right)\ln \left( {\frac{P_e }{1 - P_e }} \right),} &
{s = 0,1} \\
\end{cases}
\]
\end{corollary}

\begin{theorem} We have
\[
\xi _s (p_1 ,p_2 ) \ge \xi _s (P_e ),
\]

\noindent where
\[
\xi _s (P_e ) = \begin{cases}
 {IT_s (P_e ) = \left[ {s(s - 1)} \right]^{ - 1}\left[ {\left( {\frac{P_e ^s
+ \left( {1 - P_e } \right)^s}{2}} \right)\left( {\frac{1}{2}} \right)^{1 -
s} - \frac{1}{2}} \right],} & {s \ne 0,1} \\
 {I(P_e ) = \frac{1}{2}\left[ {P_e \ln \left( {\frac{2P_e }{2P_e - 1}}
\right)\left. { + \,\,\left( {1 - P_e } \right)\ln \left( {\frac{2\left( {1
- P_e } \right)}{2P_e - 1}} \right)} \right],} \right.} & {s = 1} \\
 {T(P_e ) = \left( {\frac{2P_e - 1}{2}} \right)\ln \left( {\frac{2P_e -
1}{2\sqrt {P_e \left( {1 - P_e } \right)} }} \right),} & {s = 0} \\
\end{cases}.
\]
\end{theorem}

\begin{proof}  Following the lines of Theorem 2.1, we can write
\begin{equation}
\label{eq12}
\xi _s (x) = \begin{cases}
 {IT_s (x) = \left[ {s(s - 1)} \right]^{ - 1}\left[ {\left( {\frac{P_e (x)^s
+ \left( {1 - P_e (x)} \right)^s}{2}} \right)\left( {\frac{1}{2}} \right)^{1
- s} - \frac{1}{2}} \right],} & {s \ne 0,1} \\
 {I(x) = \frac{1}{2}\left[ {P_e (x)\ln \left( {\frac{2P_e (x)}{2P_e (x) -
1}} \right) + \left( {1 - P_e (x)} \right)\ln \left( {\frac{2\left( {1 - P_e
(x)} \right)}{2P_e (x) - 1}} \right)} \right]} & {s = 1} \\
 {T(x) = \left( {\frac{2P_e (x) - 1}{2}} \right)\ln \left( {\frac{2P_e (x) -
1}{2\sqrt {P_e (x)\left( {1 - P_e (x)} \right)} }} \right),} & {s = 0} \\
\end{cases}.
\end{equation}

Now consider a function
\[
g_s (a) = \begin{cases}
 {IT_s (a) = \left[ {s(s - 1)} \right]^{ - 1}\left[ {\left( {\frac{a^s +
\left( {1 - a} \right)^s}{2}} \right)\left( {\frac{1}{2}} \right)^{1 - s} -
\frac{1}{2}} \right],} & {s \ne 0,1} \\
 {I(a) = \frac{1}{2}\left[ {a\ln \left( {\frac{2a}{2a - 1}} \right) + \left(
{1 - a} \right)\ln \left( {\frac{2\left( {1 - a} \right)}{2a - 1}} \right)}
\right],} & {s = 1} \\
 {T(a) = \left( {\frac{2a - 1}{2}} \right)\ln \left( {\frac{2a - 1}{2\sqrt
{a\left( {1 - a} \right)} }} \right),} & {s = 0} \\
\end{cases},
\]

\noindent
for all $0 < a < \textstyle{1 \over 2}$.

\smallskip
We have
\[
{g}'_s (a) = \begin{cases}
 {IT_s ^\prime (a) = \frac{2^{s - 2}\left[ {a^s(a - 1) + a(1 - a)^s}
\right]}{(s - 1)(a - 1)},} & {s \ne 0,1} \\
 {{I}'(a) = \ln \left( {\frac{a}{1 - a}} \right),} & {s = 1} \\
 {{T}'(a) = \frac{2a - 1}{4a(1 - a)},} & {s = 0} \\
\end{cases}
\]

\noindent and
\[
{g}''_s (a) = \begin{cases}
 {I{T}''_s (a) = 2^{s - 2}\left[ {\frac{a^s(1 - a)^2 + a^2(1 - a)^s}{a^2(1 -
a)^2}} \right],} & {s \ne 0,1} \\
 {{I}''(a) = \frac{1}{2a(1 - a)},} & {s = 1} \\
 {{T}''(a) = \frac{a^2 + (1 - a)^2}{4a^2(1 - a)^2},} & {s = 0} \\
\end{cases}.
\]

This gives ${g}''_s (a) \ge 0$, i.e, $g_s (a)$ is convex function of $a$ ($0
< a < \textstyle{1 \over 2})$ for any $s \in ( - \infty ,\infty )$.

\smallskip
Taking expectation on both sides of (\ref{eq12}) and using Jensen's inequality we
get the required result.
\end{proof}

\begin{corollary}  When $p_1 = p_2 = \frac{1}{2}$, we have
\[
\xi _s \ge \begin{cases}
 {2\left[ {s(s - 1)} \right]^{ - 1}\left[ {\left( {\frac{P_e ^s + \left( {1
- P_e } \right)^s}{2}} \right)\left( {\frac{1}{2}} \right)^{1 - s} -
\frac{1}{2}} \right],} & {s \ne 0,1} \\
 {P_e \ln \left( {\frac{2P_e }{2P_e - 1}} \right) + \,\,\left( {1 - P_e }
\right)\ln \left( {\frac{2\left( {1 - P_e } \right)}{2P_e - 1}} \right),} &
{s = 1} \\
 {\left( {2P_e - 1} \right)\ln \left( {\frac{2P_e - 1}{2\sqrt {P_e \left( {1
- P_e } \right)} }} \right),} & {s = 0} \\
\end{cases}
\]
\end{corollary}

\section{$f-$Divergence and Probability of Error}

Csisz\'{a}r \cite{csi} have given a measure for the divergence between two
probability density functions, say $p(x)$ and $q(x)$. This so called $f -$divergence given by
\begin{equation}
\label{eq13}
C_f (p,q) = \int\limits_{\rm X} {f\left( {\frac{p(x)}{q(x)}} \right)q(x)dx} .
\end{equation}

The function $f(x)$, with $x \in (0,\infty )$ is a convex function which has
to satisfy the conditions
\begin{equation}
\label{eq14}
f(0) = \mathop {\lim }\limits_{u \downarrow 0} f(u);
\,
0f\left( {\frac{0}{0}} \right) = 0 \, ;
\,
0f\left( {\frac{a}{0}} \right) = \mathop {\lim }\limits_{ \in \downarrow
\infty } \in f\left( {\frac{a}{ \in }} \right) = a\mathop {\lim }\limits_{x
\to \infty } \frac{f(u)}{u}.
\end{equation}

It can be easily checked that $C_f (p,q) \ge f(1)$ and that $C_f (p,q) =
f(1)$ only when $p(x) = q(x)$ a.e. Thus, $C_f (p,q) - f(1)$ is a distance or
divergence measure in the sense that $C_f (p,q) - f(1) \ge 0$. However, it
is not symmetric in $p$ and $q$ and in general does not satisfy triangle
inequality. Some applications of $f-$divergence in connection with divergence
measures are given in Taneja and Kumar \cite{tak}.

\smallskip
Boekee and Van der Lubbe \cite{bov} have introduced the average $f -
$divergence between two hypothesis $C_1 $ and $C_2 $ in terms of their ``a
posteriori'' probabilities. This average $f - $divergence is defined as
\begin{equation}
\label{eq15}
\overline C _f (C_1 ,C_2 ) = \int\limits_{\rm X} {f\left( {\frac{P(C_1 \vert
x)}{P(C_2 \vert x)}} \right)P(C_1 \vert x)p(x)dx}
 = E_X \left\{ {f\left( {\frac{P(C_1 \vert x)}{P(C_2 \vert x)}} \right)P(C_1
\vert x)} \right\}.
\end{equation}

If we introduce the function $f_ * (u) = u\,f\left( {\frac{1 - u}{u}} \right)$ and set $u = u(x) = P(C_2 \vert x)$, it is easy to see from $P(C_1 \vert x) = 1 - P(C_2 \vert x)$ that
\[
\overline C _f (C_1 ,C_2 ) = E_X \left\{ {f_ * \left( {\frac{P(C_1 \vert
x)}{P(C_2 \vert x)}} \right)P(C_1 \vert x)} \right\}
\]

From Vajda \cite{vaj} it follows that $f_ * (u)$ is convex on $[0,1]$ and is strictly convex iff $f(u)$ is strictly convex.

\subsection{A Class of Upper Bounds}

In \cite{bov} it has been shown that the Bayesian probability of error can be upper bounds in terms of the average $f - $divergence $\overline C _f (C_1, C_2 )$. This upper bound is given by
\begin{equation}
\label{eq16}
P_e \le \frac{f_0 P(C_2 ) + f_\infty P(C_1 ) - \overline C _f (C_1 ,C_2
)}{f_2 - f_1 },
\end{equation}

\noindent
where $f_2 $ should be finite with $f_0 = \mathop {\lim }\limits_{u \downarrow \infty } f(u); \,
f_1 = f(1); \, f_2 = f_0 + f_\infty ; \, f_\infty = \mathop {\lim }\limits_{u \to \infty } \frac{f(u)}{u}.
$

\smallskip
The above bound is valid only for \textit{every convex function} $f(u)$ which satisfies the conditions
given in (\ref{eq14}). However, if $f_ * (u) = u\,f\left( {\frac{1 - u}{u}} \right)$ is symmetric with respect to $u = \frac{1}{2}$ i.e., $f_ * (u) = f_ * (1 - u)$, $\forall u \in (0,1)$, the bound (\ref{eq16}) can be written in a simpler form given by
\begin{equation}
\label{eq17}
P_e \le \frac{1}{2f_\infty - f_1 }\left[ {f_\infty - \overline C _f (C_1
,C_2 )} \right],
\end{equation}

\noindent
provided $f_\infty $ is finite and $f_ * (u) = f_ * (1 - u)$, $\forall u \in
(0,1)$.

\smallskip
As a consequence of result (\ref{eq17}) we have the following two theorems in terms
of measures given in (\ref{eq1}) and (\ref{eq4}) respectively.

\begin{theorem} The following bound holds
\begin{equation}
\label{eq18}
P_e \le \frac{1}{2}\left[ {1 + s(s - 1)\overline \zeta _s (C_1 ,C_2 )}
\right],
\,
0 < s < 1,
\end{equation}

\noindent where
\[
\overline \zeta _s (C_1 ,C_2 ) = E_X \left\{ {f_{_{\zeta _s } }^ * \left(
{\frac{P(C_1 \vert x)}{P(C_2 \vert x)}} \right)P(C_1 \vert x)} \right\}
\]

\noindent with
\[
f_{_{\zeta _s } }^ * (x) = \begin{cases}
 {\left[ {s(s - 1)} \right]^{ - 1}\left[ {(1 - x)^{1 - s}x^s + x^{1 - s}(1 -
x)^s - 1} \right],} & {s \ne 0,1} \\
 {(2x - 1)\ln \left( {\frac{x}{1 - x}} \right)} & {s = 0,1} \\
\end{cases}
\]

\noindent
for all$\;x \in (0,1)$.
\end{theorem}

\begin{proof} Let us consider
\[
\overline \zeta _s (C_1 ,C_2 ) = E_X \left\{ {f_{_{\zeta _s } }^ * \left(
{\frac{P(C_1 \vert x)}{P(C_2 \vert x)}} \right)P(C_1 \vert x)} \right\},
\]

\noindent where
\begin{equation}
\label{eq19}
f_{\zeta _s }^ * (x) = \begin{cases}
 {x\,f_{J_s } \left( {\frac{1 - x}{x}} \right) = \frac{(1 - x)^{1 - s}x^s +
x^{1 - s}(1 - x)^s - 1}{s(s - 1)},} & {s \ne 0,1} \\
 {x\,f_J \left( {\frac{1 - x}{x}} \right) = (2x - 1)\ln \left( {\frac{x}{1 -
x}} \right),} & {s = 0,1} \\
\end{cases}
\end{equation}

\noindent with
\begin{equation}
\label{eq20}
f_{\zeta _s } (x) = \begin{cases}
 {f_{J_s } (x) = \left[ {s(s - 1)} \right]^{ - 1}\left[ {x^s + x^{1 - s} -
(x + 1)} \right],} & {s \ne 0,1} \\
 {f_J (x) = (x - 1)\ln x} & {s = 0,1} \\
\end{cases}
\end{equation}

\noindent
for all $x \in (0,1)$.

\smallskip
The convexity of the function $f_{\zeta _s } (x)$ given (\ref{eq20}) can be seen in
Taneja \cite{tan4}. From the expression (\ref{eq19}), we observe that
\begin{equation}
\label{eq21}
f_{\zeta _s }^ * (x) = f_{\zeta _s }^ * (1 - x).
\end{equation}

Also, we have
\begin{equation}
\label{eq22}
f_{\zeta _\infty } = \mathop {\lim }\limits_{x \to \infty } \frac{f_\zeta
(x)}{x} = \begin{cases}
 {\frac{ - 1}{s(s - 1)},} & {0 < s < 1} \\
 \infty & {\mbox{otherwise}} \\
\end{cases}
\end{equation}

\noindent and
\begin{equation}
\label{eq23}
f_\zeta (1) = 0.
\end{equation}

Expression (\ref{eq17}) together with (\ref{eq21})-(\ref{eq23}) give the required result (\ref{eq18}).
\end{proof}

\begin{theorem} The following bound holds
\begin{equation}
\label{eq24}
P_e \le \frac{1}{2}\left[ {1 - \frac{2s(s - 1)}{2^{ - s} - 1}\overline \xi
_s (C_1 ,C_2 )} \right], \, s < 1
\end{equation}

\noindent where
\[
\overline \xi _s (C_1 ,C_2 ) = E_X \left\{ {f_{\xi _s }^ * \left(
{\frac{P(C_1 \vert x)}{P(C_2 \vert x)}} \right)P(C_1 \vert x)} \right\}
\]

\noindent with
\[
f_{\xi _s }^ * (x) = \begin{cases}
 {f_{IT_s }^\ast (x) = \left[ {2s(s - 1)} \right]^{ - 1}\left[ {2^{ -
s}\left( {(1 - x)^{1 - s} + x^{1 - s}} \right) - 1} \right],} & {s \ne 0,1}
\\
 {f_I^\ast (x) = \frac{1}{2}\left[ {\ln 2 + x\ln x + (1 - x)\ln (1 - x)}
\right],} & {s = 0} \\
 {f_T^\ast (x) = \frac{1}{2}\ln \left( {\frac{1}{2\sqrt {x(1 - x)} }}
\right),} & {s = 1} \\
\end{cases}
\]

\noindent for all$\;x \in (0,1)$.
\end{theorem}

\begin{proof} Let us write
\[
\overline \xi _s (C_1 ,C_2 ) = E_X \left\{ {f_{\xi _s }^ * \left(
{\frac{P(C_1 \vert x)}{P(C_2 \vert x)}} \right)P(C_1 \vert x)} \right\},
\]

\noindent where
\begin{equation}
\label{eq25}
f_{\xi _s }^ * (x) = \begin{cases}
 {x\,f_{IT_s } \left( {\frac{1 - x}{x}} \right) = \frac{2^{ - s}\left( {(1 -
x)^{1 - s} + x^{1 - s}} \right) - 1}{2s(s - 1)},} & {s \ne 0,1} \\
 {x\,f_I \left( {\frac{1 - x}{x}} \right) = \frac{1}{2}\left[ {\ln 2 + x\ln
x + (1 - x)\ln (1 - x)} \right],} & {s = 0} \\
 {x\,f_T \left( {\frac{1 - x}{x}} \right) = \frac{1}{2}\ln \left(
{\frac{1}{2\sqrt {x(1 - x)} }} \right),} & {s = 1} \\
\end{cases}
\end{equation}

\noindent with
\begin{equation}
\label{eq26}
f_{\xi _s } (x) = \begin{cases}
 {f_{IT_s } (x) = \left[ {s(s - 1)} \right]^{ - 1}\left[ {\left( {\frac{x^{1
- s} + 1}{2}} \right)\left( {\frac{x + 1}{2}} \right)^s - \left( {\frac{x +
1}{2}} \right)} \right],} & {s \ne 0,1} \\
 {f_I (x) = \frac{x}{2}\ln x + \left( {\frac{x + 1}{2}} \right)\ln \left(
{\frac{x + 1}{2}} \right),} & {s = 0} \\
 {f_T (x) = \frac{x + 1}{2}\ln x\left( {\frac{x + 1}{2\sqrt x }} \right),} &
{s = 1} \\
\end{cases}
\end{equation}

\noindent
for all $\;x \in (0,1)$.

\smallskip
The convexity of the function $f_{\xi _s } (x)$ given (\ref{eq26}) can be seen in
Taneja \cite{tan3}. From the expression (\ref{eq25}) we observe that
\begin{equation}
\label{eq27}
f_{\xi _s }^ * (x) = f_{\xi _s }^ * (1 - x).
\end{equation}

Let us calculate now $f_{\xi _\infty } $. We have
\begin{align}
\label{eq28}
\mathop {\lim }\limits_{x \to \infty } \frac{f_{IT_s } (x)}{x}  &= \mathop
{\lim }\limits_{x \to \infty } \left[ {s(s - 1)} \right]^{ -
1}\frac{1}{x}\left[ {\left( {\frac{x^{1 - s} + 1}{2}} \right)\left( {\frac{x
+ 1}{2}} \right)^s - \left( {\frac{x + 1}{2}} \right)} \right]\notag\\
& = \mathop {\lim }\limits_{x \to \infty } \left[ {s(s - 1)} \right]^{ -
1}\left[ {\left( {\frac{x^{1 - s} + 1}{2x}} \right)\left( {\frac{x + 1}{2}}
\right)^s - \left( {\frac{1}{2} + \frac{1}{x}} \right)} \right] \notag\\
 &= \mathop {\lim }\limits_{x \to \infty } \left[ {s(s - 1)} \right]^{ -
1}\left[ {2^{ - 1 - s}\left( {\frac{x^{1 - s} + 1}{x}} \right)\left( {x + 1}
\right)^s - \left( {\frac{1}{2} + \frac{1}{x}} \right)} \right] \notag\\
& = \mathop {\lim }\limits_{x \to \infty } \left[ {s(s - 1)} \right]^{ -
1}\left[ {2^{ - 1 - s}\left( {\left( {\frac{x + 1}{x}} \right)^s + \frac{(x
+ 1)^s}{x}} \right) - \left( {\frac{1}{2} + \frac{1}{x}} \right)} \right] \notag\\
& = \left[ {s(s - 1)} \right]^{ - 1}\left[ {2^{ - 1 - s}\left( {\mathop {\lim
}\limits_{x \to \infty } \left( {1 + \frac{1}{x}} \right)^s + \mathop {\lim
}\limits_{x \to \infty } \frac{(x + 1)^s}{x}} \right) - \mathop {\lim
}\limits_{x \to \infty } \left( {\frac{1}{2} + \frac{1}{x}} \right)} \right] \notag \\
& = \left[ {2s(s - 1)} \right]^{ - 1}\left[ {2^{ - s}\left( {1 + \mathop
{\lim }\limits_{x \to \infty } \frac{(x + 1)^s}{x}} \right) - 1} \right],\;s
\ne 0,\,1.
 \end{align}

We can easily verify that
\begin{equation}
\label{eq29}
\mathop {\lim }\limits_{x \to \infty } \frac{(x + 1)^s}{x}\begin{cases}
 {0,} & {s < 1} \\
 {1,} & {s = 1} \\
 {\infty ,} & {s > 1} \\
\end{cases}.
\end{equation}

Expressions (\ref{eq28}) and (\ref{eq29}) together give
\begin{equation}
\label{eq30}
f_{\xi _\infty } = \mathop {\lim }\limits_{x \to \infty } \frac{f_\xi
(x)}{x} = \begin{cases}
 {\left[ {2s(s - 1)} \right]^{ - 1}\left[ {2^{ - s} - 1} \right],} & {s <
1,\;s \ne 0} \\
 {\frac{1}{2}\ln 2,} & {s = 0} \\
 {\infty ,} & {s \ge 1} \\
\end{cases}
\end{equation}

Also we have
\begin{equation}
\label{eq31}
f_\xi (1) = 0.
\end{equation}

Expression (\ref{eq17}) together with (\ref{eq27}), (\ref{eq30}) and (\ref{eq31}) give the required
result (\ref{eq24}).
\end{proof}

\smallskip
Here below are the particular cases of the Theorems 3.1 and 3.2. These can
also be obtained directly using result (\ref{eq17}). These examples are the seven
measures appearing (\ref{eq7}).

\begin{example} Let us consider the measure
\[
\overline J (C_1 ,C_2 ) = E_X \left\{ {f_{_J }^ * \left( {\frac{P(C_1 \vert
x)}{P(C_2 \vert x)}} \right)P(C_1 \vert x)} \right\},
\]

\noindent where
\[
f_J^ * (x) = (2x - 1)\ln \left( {\frac{x}{1 - x}} \right),
\,
\forall x \in (0,1).
\]

We can't obtain the upper bound on the probability of error in terms of
\textit{J-divergence}, since according to (\ref{eq22}), $f_{J_\infty } $ is infinite.
\end{example}

\begin{example}  In particular for $s = \frac{1}{2}$ in (\ref{eq18}), we have
the following upper bound on the probability of error in terms of the
\textit{Hellinger distance}:
\[
P_e \le \frac{1}{2}\left[ {1 - 2\,\overline h (C_1 ,C_2 )} \right],
\]

\noindent where
\[
\overline h (C_1 ,C_2 ) = E_X \left\{ {f_h^ * \left( {\frac{P(C_1 \vert
x)}{P(C_2 \vert x)}} \right)P(C_1 \vert x)} \right\},
\]

\noindent with
\[
f_h^ * (x) = \frac{1 - 2\sqrt {x(1 - x)} }{2},
\,
\forall x \in (0,1).
\]
\end{example}

\begin{example}  In particular for $s = 0$ in (\ref{eq24}), we have the
following upper bound on the probability of error in terms of the
\textit{Jensen-Shannon divergence}:
\[
P_e \le \frac{1}{2}\left[ {1 - \frac{2}{\ln 2}\overline I (C_1 ,C_2 )}
\right],
\]

\noindent where
\[
\overline I (C_1 ,C_2 ) = E_X \left\{ {f_I^ * \left( {\frac{P(C_1 \vert
x)}{P(C_2 \vert x)}} \right)P(C_1 \vert x)} \right\},
\]

\noindent with
\[
f_I^ * (x) = \frac{1}{2}\left[ {\ln 2 + x\ln x + (1 - x)\ln (1 - x)}
\right],
\,
\forall x \in (0,1).
\]
\end{example}

\begin{example}  Let us consider the measure
\[
\overline T (C_1 ,C_2 ) = E_X \left\{ {f_{_T }^ * \left( {\frac{P(C_1 \vert
x)}{P(C_2 \vert x)}} \right)P(C_1 \vert x)} \right\},
\]

\noindent where
\[
f_T^ * (x) = \frac{1}{2}\ln \left( {\frac{1}{2\sqrt {x(1 - x)} }} \right),
\,
\forall x \in (0,1).
\]

We can't obtain the upper bound on the probability of error in terms of the \textit{arithmetic and geometric mean divergence,} since according to (\ref{eq30}), $f_{T_\infty } $ is infinite in this case.
\end{example}

\begin{example}  In particular for $s = \frac{1}{2}$ in (\ref{eq24}), we have
the following upper bound on the probability of error in terms of \textit{d-divergence}:
\[
P_e \le \frac{1}{2}\left[ {1 - \left( {\frac{4}{2 - \sqrt 2 }}
\right)\overline d (C_1 ,C_2 )} \right],
\]

\noindent where
\[
\overline d (C_1 ,C_2 ) = E_X \left\{ {f_d^ * \left( {\frac{P(C_1 \vert
x)}{P(C_2 \vert x)}} \right)P(C_1 \vert x)} \right\},
\]

\noindent with
\[
f_d^ * (x) = \frac{1}{2} - \sqrt 2 \left( {\sqrt x + \sqrt {1 - x} }
\right),
\,
\forall x \in (0,1).
\]
\end{example}

\begin{example}  In particular for $s = - 1$ in (\ref{eq24}), we have the
following upper bound on the probability of error in terms of \textit{triangular discrimination}:
\[
P_e \le \frac{1}{2}\left[ {1 - \overline \Delta (C_1 ,C_2 )} \right],
\]

\noindent where
\[
\overline \Delta (C_1 ,C_2 ) = E_X \left\{ {f_\Delta ^ * \left( {\frac{P(C_1
\vert x)}{P(C_2 \vert x)}} \right)P(C_1 \vert x)} \right\}
\]

\noindent with
\[
f_\Delta ^ * (x) = \left( {2x - 1} \right)^2,
\,
\forall x \in (0,1).
\]
\end{example}

\begin{example}  Let us consider the measure
\[
\overline \Psi (C_1 ,C_2 ) = E_X \left\{ {f_{_\Psi }^ * \left( {\frac{P(C_1
\vert x)}{P(C_2 \vert x)}} \right)P(C_1 \vert x)} \right\},
\]

\noindent where
\[
f_\Psi ^ * (x) = \frac{\left( {2x - 1} \right)^2}{x(1 - x)},
\,
\forall x \in (0,1).
\]

We can't obtain the upper bound on the probability of error in terms of \textit{symmetric} $\chi ^2 - $\textit{divergence}, since according to (\ref{eq23}), $f_{\Psi _\infty } $is infinite for $s = 2$.
\end{example}

\section{Difference of Divergence Measures}

From the Examples 3.1-3.7 we observe that there are only four measures that give upper bounds on the probability of error. According to (\ref{eq1}), these four measures bear the following inequality:
\begin{equation}
\label{eq32}
\frac{1}{4}\Delta (P\vert \vert Q) \le I(P\vert \vert Q) \le h(P\vert \vert
Q) \le 4\,d(P\vert \vert Q).
\end{equation}

The inequality (\ref{eq32}) admits the following nonnegative differences
\begin{equation}
\label{eq33}
D_{d\Delta } (P\vert \vert Q) = 4\,d(P\vert \vert Q) - \frac{1}{4}\Delta
(P\vert \vert Q),
\end{equation}
\begin{equation}
\label{eq34}
D_{dh} (P\vert \vert Q) = 4\,d(P\vert \vert Q) - h(P\vert \vert Q),
\end{equation}
\begin{equation}
\label{eq35}
D_{dI} (P\vert \vert Q) = 4\,d(P\vert \vert Q) - I(P\vert \vert Q),
\end{equation}
\begin{equation}
\label{eq36}
D_{hI} (P\vert \vert Q) = h(P\vert \vert Q) - I(P\vert \vert Q),
\end{equation}
\begin{equation}
\label{eq37}
D_{h\Delta } (P\vert \vert Q) = h(P\vert \vert Q) - \frac{1}{4}\Delta
(P\vert \vert Q)
\end{equation}

\noindent and
\begin{equation}
\label{eq38}
D_{I\Delta } (P\vert \vert Q) = I(P\vert \vert Q) - \frac{1}{4}\Delta
(P\vert \vert Q).
\end{equation}

The above six measures can be related by the following (ref. Taneja \cite{tan6}):
\begin{align}
\label{eq39}
D_{I\Delta } (P\vert \vert Q) &\le \frac{2}{3}D_{h\Delta } (P\vert \vert Q)
\le \frac{8}{15}D_{d\Delta } (P\vert \vert Q) \le\notag\\
 \le & \frac{8}{3}D_{dh} (P\vert \vert Q) \le \frac{8}{7}D_{dI} (P\vert \vert
Q) \le 2\,D_{hI} (P\vert \vert Q).
\end{align}

Here below we shall give bounds on the probability of error in terms of the
measures (\ref{eq33})-(\ref{eq38}).

\begin{result} The following bound holds
\begin{equation}
\label{eq40}
P_e \le \frac{1}{2}\left[ {1 - \frac{4}{7 - 4\sqrt 2 }\overline D _{d\Delta
} (C_1 ,C_2 )} \right],
\end{equation}

\noindent where
\[
\overline D _{d\Delta } (C_1 ,C_2 ) = E_X \left\{ {f_{d\Delta }^ * \left(
{\frac{P(C_1 \vert x)}{P(C_2 \vert x)}} \right)P(C_1 \vert x)} \right\}
\]

\noindent with
\[
f_{d\Delta }^ * (x) = \frac{7}{4} - \sqrt 2 \left( {\sqrt x + \sqrt {1 - x}
} \right) + x(1 - x),
\,
\forall x \in (0,1).
\]
\end{result}

\begin{proof} Let us consider
\[
\overline D _{d\Delta } (C_1 ,C_2 ) = E_X \left\{ {f_{d\Delta }^ * \left(
{\frac{P(C_1 \vert x)}{P(C_2 \vert x)}} \right)P(C_1 \vert x)} \right\},
\]

\noindent where
\[
f_{d\Delta }^ * (x) = x\,f_{d\Delta } \left( {\frac{1 - x}{x}} \right)
 = \frac{7}{4} - \sqrt 2 \left( {\sqrt x + \sqrt {1 - x} } \right) + x(1 -
x),
\]

\noindent with
\[
f_{d\Delta } (x) = 2x + 2 - \left( {\sqrt x + 1} \right)\sqrt {2x + 2} -
\frac{(x - 1)^2}{4(x + 1)^2}.
\]

Now we shall prove the convexity of the function $f_{d\Delta } (x)$. We have
\[
{f}'_{d\Delta } (x) = 2 - \frac{\sqrt {2x + 2} }{2\sqrt x } - \frac{1 +
\sqrt x }{\sqrt {2x + 2} } - \frac{x - 1}{2(x + 1)^2} + \frac{(x - 1)^2}{4(x
+ 1)^2}
\]

\noindent and
\begin{align}
\label{eq41}
{f}''_{d\Delta } (x) & = \frac{(x^{3 / 2} + 1)(x + 1)^2 - 4x^{3 / 2}\sqrt {2x
+ 2} }{2x^{3 / 2}(x + 1)^3\sqrt {2x + 2} }\notag\\
&  = \frac{8}{2x^{3 / 2}(x + 1)^3\sqrt {2x + 2} }\sqrt {\frac{x + 1}{2}}
\left[ {\left( {\frac{x^{3 / 2} + 1}{2}} \right)\left( {\frac{x + 1}{2}}
\right)^{3 / 2} - x^{3 / 2}} \right], \, \forall x > 0.
\end{align}

Now we shall prove the non-negativity of the expression (\ref{eq41}). We can
easily check that
\begin{equation}
\label{eq42}
\left( {\frac{\sqrt x + 1}{2}} \right)^3 \le \frac{x^{3 / 2} + 1}{2}.
\end{equation}

On the other side we know that \cite{tan3} pp. 279:
\begin{equation}
\label{eq43}
x^{3 / 2} \le \left( {\frac{\sqrt x + 1}{2}} \right)^3\left( {\frac{x +
1}{2}} \right)^{3 / 2}.
\end{equation}

Expressions (\ref{eq42}) and (\ref{eq43}) together give
\begin{equation}
\label{eq44}
x^{3 / 2} \le \left( {\frac{\sqrt x + 1}{2}} \right)^3\left( {\frac{x +
1}{2}} \right)^{3 / 2} \le \left( {\frac{x^{3 / 2} + 1}{2}} \right)\left(
{\frac{x + 1}{2}} \right)^{3 / 2}.
\end{equation}

Expression (\ref{eq44}) proves the non-negativity of the expression (\ref{eq41}), i.e.,
${f}''_{d\Delta } (x) \ge 0$, $\forall x > 0$. This proves the convexity of
the function $f_{d\Delta } (x)$, $\forall x > 0$.

\smallskip
We have
\begin{align}
\label{eq45}
f_{d\Delta }^ * (x) & = f_{d\Delta }^ * (1 - x), \, x \in (0,1),\\
\label{eq46}
f_{\left( {d\Delta } \right)_\infty } & = \mathop {\lim }\limits_{x \to \infty
} \frac{f_{d\Delta } (x)}{x} = \frac{7}{4} - \sqrt 2 = \frac{7 - 4\sqrt 2
}{4}
\intertext{and}
\label{eq47}
f_1 &= f_{d\Delta } (1) = 0.
\end{align}

The result (\ref{eq17}) together with the expressions (\ref{eq45})-(\ref{eq47}) gives the
required result (\ref{eq40}).
\end{proof}

\begin{result}The following bound holds
\begin{equation}
\label{eq48}
P_e \le \frac{1}{2}\left[ {1 - \frac{2}{4 - 2\sqrt 2 - \ln 2}\,\overline D
_{dI} (C_1 ,C_2 )} \right],
\end{equation}

\noindent where
\[
\overline D _{dI} (C_1 ,C_2 ) = E_X \left\{ {f_{dI}^ * \left( {\frac{P(C_1
\vert x)}{P(C_2 \vert x)}} \right)P(C_1 \vert x)} \right\}
\]

\noindent with
\[
f_{dI}^ * (x) = 2 - 2\sqrt 2 \left( {\sqrt x + \sqrt {1 - x} } \right) -
\frac{1}{2}\left[ {\ln 2 + x\ln x + (1 - x)\ln (1 - x)} \right],
\,
\forall x \in (0,1).
\]
\end{result}

\begin{proof} Let us consider
\[
\overline D _{dI} (C_1 ,C_2 ) = E_X \left\{ {f_{dI}^ * \left( {\frac{P(C_1
\vert x)}{P(C_2 \vert x)}} \right)P(C_1 \vert x)} \right\},
\]

\noindent where
\begin{align}
f_{dI}^ * (x) & = x\,f_{dI} \left( {\frac{1 - x}{x}} \right) \notag\\
& = 2 - 2\sqrt 2 \left( {\sqrt x + \sqrt {1 - x} } \right) -
\frac{1}{2}\left[ {\ln 2 + x\ln x + (1 - x)\ln (1 - x)} \right]\notag
\end{align}

\noindent with
\[
f_{dI} (x) = 2x + 2 - \left( {\sqrt x + 1} \right)\sqrt {2x + 2} -
\frac{1}{2}x\ln x - \frac{x + 1}{2}\ln \left( {\frac{2}{x + 1}} \right).
\]

Now we shall prove the convexity of the function $f_{dI} (x)$. We have
\[
{f}'_{dI} (x) = 2 - \frac{\sqrt {2x + 2} }{2\sqrt x } - \frac{1 + \sqrt x
}{\sqrt {2x + 2} } - \frac{1}{2}\ln x - \frac{1}{2}\ln \left( {\frac{2}{x +
1}} \right)
\]

\noindent and
\begin{align}
\label{eq49}
{f}''_{dI} (x) & = \frac{x^{3 / 2} + 1 - \sqrt x \sqrt {2x + 2} }{2x^{3 / 2}(x
+ 1)\sqrt {2x + 2} }\notag\\
& = \frac{1}{x^{3 / 2}(x + 1)\sqrt {2x + 2} }\left[ {\left( {\frac{x^{3 / 2}
+ 1}{2}} \right) - \sqrt x \sqrt {\frac{x + 1}{2}} } \right], \, \forall x > 0.
\end{align}

Now we shall prove the non-negativity of the expression (\ref{eq49}). We know that
\cite{tan4}, pp. 358:

\begin{center}
$\sqrt x \le \left( {\frac{\sqrt x + 1}{2}} \right)^2$ and $\left( {\frac{\sqrt
x + 1}{2}} \right)\sqrt {\frac{x + 1}{2}} \le \left( {\frac{x + 1}{2}}
\right)$.
\end{center}

This give
\begin{equation}
\label{eq50}
\sqrt x \sqrt {\frac{x + 1}{2}} \le \left( {\frac{\sqrt x + 1}{2}}
\right)^2\sqrt {\frac{x + 1}{2}} \le \left( {\frac{x + 1}{2}} \right)\left(
{\frac{\sqrt x + 1}{2}} \right).
\end{equation}

By simple calculations, we can check that
\begin{equation}
\label{eq51}
\left( {\frac{x + 1}{2}} \right)\left( {\frac{\sqrt x + 1}{2}} \right) \le
\frac{x^{3 / 2} + 1}{2}.
\end{equation}

The expressions (\ref{eq50}) and (\ref{eq51}) together give the non-negativity of the
expression (\ref{eq49}), i.e., ${f}''_{dI} (x) \ge 0$, $\forall x > 0$. This
proves the convexity of the function $f_{dI} (x)$, $\forall x > 0$.

\smallskip
We have
\begin{align}
\label{eq52}
f_{_{d\Delta } }^ * (x) & = f_{_{d\Delta } }^ * (1 - x), \, x \in (0,1),\\
\label{eq53}
f_{\left( {d\Delta } \right)_\infty } & = \mathop {\lim }\limits_{x \to \infty
} \frac{f_{d\Delta } (x)}{x} = 2 - \sqrt 2 - \frac{1}{2}\ln 2 = \frac{2}{4 - 2\sqrt 2 - \ln 2}
\intertext{and}
\label{eq54}
f_1 &= f_{d\Delta } (1) = 0.
\end{align}

The result (\ref{eq17}) together with the expressions (\ref{eq52})-(\ref{eq54}) gives the
required result (\ref{eq48}).
\end{proof}

\begin{result} The following bound holds
\begin{equation}
\label{eq55}
P_e \le \frac{1}{2}\left[ {1 - \frac{2}{3 - 2\sqrt 2 }\,\overline D _{dh}
(C_1 ,C_2 )} \right],
\end{equation}

\noindent where
\[
\overline D _{dh} (C_1 ,C_2 ) = E_X \left\{ {f_{dh}^ * \left( {\frac{P(C_1
\vert x)}{P(C_2 \vert x)}} \right)P(C_1 \vert x)} \right\}
\]

\noindent with
\[
f_{dh}^ * (x) = \frac{3}{2} - \sqrt 2 \left( {\sqrt x + \sqrt {1 - x} }
\right) + \frac{1}{2}\sqrt {x(1 - x)} ,
\,
\forall x \in (0,1).
\]
\end{result}

\begin{proof} Let us consider
\[
\overline D _{dh} (C_1 ,C_2 ) = E_X \left\{ {f_{dh}^ * \left( {\frac{P(C_1
\vert x)}{P(C_2 \vert x)}} \right)P(C_1 \vert x)} \right\},
\]

\noindent where
\[
f_{dh}^ * (x) = x\,f_{dh} \left( {\frac{1 - x}{x}} \right)
 = \frac{3}{2} - \sqrt 2 \left( {\sqrt x + \sqrt {1 - x} } \right) +
\frac{1}{2}\sqrt {x(1 - x)}
\]

\noindent with
\[
f_{dh} (x) = 2x + 2 - \left( {\sqrt x + 1} \right)\sqrt {2x + 2} -
\frac{\left( {\sqrt x - 1} \right)^2}{2}.
\]

Now we shall prove the convexity of the function $f_{dh} (x)$. We have
\[
{f}'_{dh} (x) = 2 - \frac{\sqrt {2x + 2} }{2\sqrt x } - \frac{1 + \sqrt x
}{\sqrt {2x + 2} } - \frac{\sqrt x - 1}{2\sqrt x }
\]

\noindent and
\begin{align}
{f}''_{dh} (x) & = \frac{2(x^{3 / 2} + 1) - (x + 1)\sqrt {2x + 2} }{4x^{3 /
2}(x + 1)\sqrt {2x + 2} }\notag\\
\label{eq56}
& = \frac{1}{x^{3 / 2}(x + 1)\sqrt {2x + 2} }\left[ {\frac{x^{3 / 2} + 1}{2}
- \left( {\frac{x + 1}{2}} \right)^{3 / 2}} \right], \, \forall x > 0,
\end{align}

The non-negativity of the expression (\ref{eq56}) follows from the fact that the function $\left( {\frac{x^s + 1}{2}} \right)^{1 / s},$  $s \ne 0$ is monotonically increasing function of $s$. For proof refer to Beckenbach and Bellman \cite{beb}. Thus, we have ${f}''_{dh} (x) \ge 0$, $\forall x > 0$, consequently proving the convexity of the function $f_{dh} (x)$, $\forall x > 0$.

\smallskip
We have
\begin{align}
\label{eq57}
f_{dh}^ * (x) & = f_{dh}^ * (1 - x), \, x \in (0,1),\\
\label{eq58}
f_{\left( {dh} \right)_\infty } & = \mathop {\lim }\limits_{x \to \infty }
\frac{f_{dh} (x)}{x} = \frac{3}{2} - \sqrt 2 = \frac{3 - 2\sqrt 2 }{2}
\intertext{and}
\label{eq59}
f_1 &= f_{dh} (1) = 0.
\end{align}

The result (\ref{eq17}) together with the expressions (\ref{eq57})-(\ref{eq59}) gives the required result (\ref{eq55}).
\end{proof}

\begin{result}The following bound holds
\begin{equation}
\label{eq60}
P_e \le \frac{1}{2}\left[ {1 - 4\,\overline D _{h\Delta } (C_1 ,C_2 )}
\right],
\end{equation}

\noindent where
\[
\overline D _{h\Delta } (C_1 ,C_2 ) = E_X \left\{ {f_{h\Delta }^ * \left(
{\frac{P(C_1 \vert x)}{P(C_2 \vert x)}} \right)P(C_1 \vert x)} \right\}
\]

\noindent with
\[
f_{h\Delta }^ * (x) = \frac{1}{4} + x(1 - x) - \sqrt {x(1 - x)} ,
\,
\forall x \in (0,1).
\]
\end{result}

\begin{proof} Let us consider
\[
\overline D _{h\Delta } (C_1 ,C_2 ) = E_X \left\{ {f_{h\Delta }^ * \left(
{\frac{P(C_1 \vert x)}{P(C_2 \vert x)}} \right)P(C_1 \vert x)} \right\},
\]

\noindent where
\[
f_{h\Delta }^ * (x) = x\,f_{h\Delta } \left( {\frac{1 - x}{x}} \right)
 = \frac{1}{4} + x(1 - x) - \sqrt {x(1 - x)}
\]

\noindent with
\[
f_{h\Delta } (x) = \frac{\left( {\sqrt x - 1} \right)^2}{2} - \frac{(x -
1)^2}{4(x - 1)}.
\]

Now we shall prove the convexity of the function $f_{h\Delta } (x)$. We have
\[
{f}'_{h\Delta } (x) = \frac{\sqrt x - 1}{2\sqrt x } - \frac{x - 1}{2(x + 1)}
+ \frac{(x - 1)^2}{4(x + 1)^2}
\]

\noindent and
\begin{align}
\label{eq61}
{f}''_{h\Delta } (x) & = \frac{x^4 + 3x^2 + 3x + x - 8x^{5 / 2}}{4x^{5 / 2}(x
+ 1)^3}\notag\\
&  = \frac{(x + 1)^3 - x^{3 / 2}}{4x^{3 / 2}(x + 1)^3}\notag\\
& = \frac{\left( {\sqrt x - 1} \right)^2\left[ {\left( {\sqrt x + 1}
\right)^2\left( {x + 1} \right) + 4x} \right]}{4x^{3 / 2}(x + 1)^3}  \ge 0, \, \forall x > 0.
\end{align}

Since, ${f}''_{h\Delta } (x) \ge 0$, $\forall x > 0$ proving the convexity of the function $f_{h\Delta } (x)$, $\forall x > 0$.

\smallskip
We have
\begin{align}
\label{eq62}
f_{h\Delta }^ * (x) & = f_{h\Delta }^ * (1 - x), \, x \in (0,1),\\
\label{eq63}
f_{\left( {h\Delta } \right)_\infty }  &= \mathop {\lim }\limits_{x \to \infty
} \frac{f_{h\Delta } (x)}{x} = \frac{1}{4}
\intertext{and}
\label{eq64}
f_1 &= f_{h\Delta } (1) = 0.
\end{align}

The result (\ref{eq17}) together with the expressions (\ref{eq62})-(\ref{eq64}) gives the required result (\ref{eq60}).
\end{proof}

\begin{result} The following bound holds
\begin{equation}
\label{eq65}
P_e \le \frac{1}{2}\left[ {1 - \frac{2}{1 - \ln 2}\,\overline D _{hI} (C_1
,C_2 )} \right],
\end{equation}

\noindent where
\[
\overline D _{hI} (C_1 ,C_2 ) = E_X \left\{ {f_{hI}^ * \left( {\frac{P(C_1
\vert x)}{P(C_2 \vert x)}} \right)P(C_1 \vert x)} \right\}
\]

\noindent with
\[
f_{hI}^ * (x) = \frac{1}{2}\left[ {1 - 2\sqrt {x(1 - x)} - \left( {\ln 2 +
x\ln x + (1 - x)\ln (1 - x)} \right)} \right], \, \forall x \in (0,1).
\]
\end{result}

\begin{proof} Let us consider
\[
\overline D _{hI} (C_1 ,C_2 ) = E_X \left\{ {f_{hI}^ * \left( {\frac{P(C_1
\vert x)}{P(C_2 \vert x)}} \right)P(C_1 \vert x)} \right\},
\]

\noindent where
\begin{align}
f_{hI}^ * (x) & = x\,f_{hI} \left( {\frac{1 - x}{x}} \right)\notag\\
& = \frac{1}{2}\left[ {1 - 2\sqrt {x(1 - x)} - \left( {\ln 2 + x\ln x + (1 - x)\ln (1 - x)} \right)} \right]\notag
\end{align}

\noindent with
\[
f_{hI} (x) = \frac{\left( {\sqrt x - 1} \right)^2}{2} - \frac{x}{2}\ln x -
\frac{x + 1}{2}\ln \left( {\frac{2}{x + 1}} \right).
\]

Now we shall prove the convexity of the function $f_{hI} (x)$. We have
\[
{f}'_{hI} (x) = \frac{\sqrt x - 1}{2\sqrt x } - \frac{1}{2}\ln x -
\frac{1}{2}\ln \left( {\frac{2}{x + 1}} \right)
\]

\noindent and
\begin{equation}
\label{eq66}
{f}''_{hI} (x) = \frac{\left( {\sqrt x - 1} \right)^2}{4x^{3 / 2}(x + 1)}
\ge 0,
\,
\forall x > 0.
\end{equation}

Since, ${f}''_{hI} (x) \ge 0$, $\forall x > 0$ proving the convexity of the
function $f_{hI} (x)$, $\forall x > 0$.

\smallskip
We have
\begin{align}
\label{eq67}
f_{hI}^ * (x) & = f_{h\Delta }^ * (1 - x), \, x \in (0,1),\\
\label{eq68}
f_{\left( {hI} \right)_\infty } & = \mathop {\lim }\limits_{x \to \infty }
\frac{f_{hI} (x)}{x} = \frac{1}{2} - \frac{1}{2}\ln 2 = \frac{1 - \ln 2}{2}
\intertext{and}
\label{eq69}
f_1 & = f_{hI} (1) = 0.
\end{align}

The result (\ref{eq17}) together with the expressions (\ref{eq67})-(\ref{eq69}) gives the required result (\ref{eq65}).
\end{proof}

\begin{result}The following bound holds
\begin{equation}
\label{eq70}
P_e \le \frac{1}{2}\left[ {1 - \frac{4}{2\ln 2 - 1}\,\overline D _{I\Delta }
(C_1 ,C_2 )} \right],
\end{equation}

\noindent where
\[
\overline D _{I\Delta } (C_1 ,C_2 ) = E_X \left\{ {f_{I\Delta }^ * \left(
{\frac{P(C_1 \vert x)}{P(C_2 \vert x)}} \right)P(C_1 \vert x)} \right\}
\]

\noindent with
\[
f_{I\Delta }^ * (x) = \frac{1}{2}\left[ {\ln 2 + x\ln x + (1 - x)\ln (1 -
x)} \right] - (2x - 1)^2,
\,
\forall x \in (0,1).
\]
\end{result}

\begin{proof} Let us consider
\[
\overline D _{I\Delta } (C_1 ,C_2 ) = E_X \left\{ {f_{I\Delta }^ * \left(
{\frac{P(C_1 \vert x)}{P(C_2 \vert x)}} \right)P(C_1 \vert x)} \right\},
\]

\noindent where
\begin{align}
f_{I\Delta }^ * (x) & = x\,f_{I\Delta } \left( {\frac{1 - x}{x}} \right)\notag\\
&  = \frac{1}{2}\left[ {\ln 2 + x\ln x + (1 - x)\ln (1 - x)} \right] - (2x - 1)^2 \notag
\end{align}

\noindent with
\[
f_{I\Delta } (x) = \frac{x}{2}\ln x + \frac{x + 1}{2}\ln \left( {\frac{2}{x
+ 1}} \right) - \frac{(x - 1)^2}{4(x + 1)}.
\]

Now we shall prove the convexity of the function $f_{I\Delta } (x)$. We have
\[
{f}'_{I\Delta } (x) = \frac{1}{2}\ln x + \frac{1}{2}\ln \left( {\frac{2}{x +
1}} \right) - \frac{x - 1}{2(x + 1)} + \frac{(x - 1)^2}{4(x + 1)^2}
\]

\noindent and
\begin{equation}
\label{eq71}
{f}''_{I\Delta } (x) = \frac{(x - 1)^2}{2x(x + 1)^3} \ge 0, \, \forall x > 0.
\end{equation}

Since, ${f}''_{I\Delta } (x) \ge 0$, $\forall x > 0$ proving the convexity
of the function $f_{I\Delta } (x)$, $\forall x > 0$.

\smallskip
We have
\begin{align}
\label{eq72}
f_{I\Delta }^ * (x) & = f_{I\Delta }^ * (1 - x), \, x \in (0,1),\\
\label{eq73}
f_{\left( {I\Delta } \right)_\infty } & = \mathop {\lim }\limits_{x \to \infty
} \frac{f_{I\Delta } (x)}{x}  = \frac{1}{2}\ln 2 - \frac{1}{4} = \frac{2\ln 2 - 1}{4}
\intertext{and}
\label{eq74}
f_1 &= f_{I\Delta } (1) = 0.
\end{align}

The result (\ref{eq17}) together with the expressions (\ref{eq72})-(\ref{eq74}) gives the required result (\ref{eq70}).
\end{proof}

\subsection{Comparison Results}

(i) In view of (\ref{eq39}) and (\ref{eq71}) we have
\begin{equation}
\label{eq75}
P_e \le \frac{1}{2}\left[ {1 - \frac{8}{3}\,\overline D _{h\Delta } (C_1
,C_2 )} \right] \le \frac{1}{2}\left[ {1 - 4\,\overline D _{I\Delta } (C_1
,C_2 )} \right].
\end{equation}

Since $2\ln 2 - 1 > 0$, then from (\ref{eq70}) and (\ref{eq75}) we have
\begin{equation}
\label{eq76}
P_e \le \frac{1}{2}\left[ {1 - \frac{4}{2\ln 2 - 1}\,\overline D _{I\Delta }
(C_1 ,C_2 )} \right]
 \le \frac{1}{2}\left[ {1 - 4\,\overline D _{I\Delta } (C_1 ,C_2 )}
\right].
\end{equation}

Thus, we observe from (\ref{eq75}) that the bound obtained in (\ref{eq70}) is sharper
than the one obtained applying directly the inequality (\ref{eq39}).

\smallskip
\noindent
(ii) In view of (\ref{eq39}) and (\ref{eq65}), we have
\begin{equation}
\label{eq77}
P_e \le \frac{1}{2}\left[ {1 - \frac{2}{1 - \ln 2}\,\overline D _{hI} (C_1
,C_2 )} \right] \le \frac{1}{2}\left[ {1 - \frac{8}{15(1 - \ln 2)}D_{d\Delta
} (P\vert \vert Q)} \right].
\end{equation}

Since $3 + 8\sqrt 2 - 15\ln 2 > 0$, then from (\ref{eq50}) and (\ref{eq77}), we have
\begin{equation}
\label{eq78}
P_e \le \frac{1}{2}\left[ {1 - \frac{2}{3 - 2\sqrt 2 }\,\overline D _{dh}
(C_1 ,C_2 )} \right] \le \frac{1}{2}\left[ {1 - \frac{8}{15(1 - \ln
2)}D_{d\Delta } (P\vert \vert Q)} \right].
\end{equation}

Again, we observe from (\ref{eq78}) that the bound obtained in (\ref{eq50}) is sharper
than the one obtained directly applying the inequality (\ref{eq39}).

\smallskip
\noindent
(iii) In view of (\ref{eq32}) and (\ref{eq40}), we have
\begin{equation}
\label{eq79}
P_e \le \frac{1}{2}\left[ {1 - \,\frac{4}{7 - 4\sqrt 2 }\,\overline D
_{d\Delta } (C_1 ,C_2 )} \right] \le \frac{1}{2}\left[ {1 - \frac{4}{7 -
4\sqrt 2 }\,\overline D _{h\Delta } (C_1 ,C_2 ))} \right].
\end{equation}

Since $7 - 4\sqrt 2 > 1$, then from (\ref{eq60}) and (\ref{eq79}) we have
\begin{equation}
\label{eq80}
P_e \le \frac{1}{2}\left[ {1 - 4\,\overline D _{h\Delta } (C_1 ,C_2 )}
\right] \le \frac{1}{2}\left[ {1 - \frac{4}{7 - 4\sqrt 2 }\,\overline D
_{h\Delta } (C_1 ,C_2 ))} \right].
\end{equation}

Here also, we observe from (\ref{eq80}) that making different comparisons, i.e., instead of using (\ref{eq39}) if we use (\ref{eq32}) still the result obtained in (\ref{eq60}) is sharper. Using the similar arguments we can also make comparisons for the other bounds.

\end{document}